\documentclass[12pt, a4paper, reqno]{amsart}
\usepackage{mathrsfs}
\usepackage{bbm}
\usepackage{pgf,tikz}
\usepackage{tabu}
\usepackage{amsmath,amscd,amssymb,latexsym}
\usepackage[all]{xy}

\input xypic

\evensidemargin=\oddsidemargin
\DeclareMathVersion{can}
\DeclareMathAlphabet{\can}{OT1}{cmss}{m}{n}
\newtheorem{thm}{Theorem}[section]
\newtheorem{cor}[thm]{Corollary}
\newtheorem{lem}[thm]{Lemma}

\newtheorem{rem}[thm]{Remark}
\newtheorem{exa}[thm]{Example}
\theoremstyle{definition}

\theoremstyle{fact}

\theoremstyle{conjecture}

\numberwithin{equation}{section}

\newcommand{\Tr}{\operatorname{Tr}}

\begin{document}
\title[Optimal few-weight codes  from simplicial complexes]
{Optimal few-weight codes  from \\simplicial complexes}



\author[Y. Wu]{ Yansheng Wu}

\address{\rm Department of Mathematics, Nanjing University of Aeronautics and Astronautics, Nanjing,  211100, P. R. China; State Key Laboratory of Cryptology, P. O. Box 5159, Beijing, 100878, P. R. China.  The author is now with the Department of Mathematics, Ewha Womans
University, 52, Ewhayeodae-gil, Seodaemun-gu, Seoul, 03760, South Korea.}
\email{wysasd@163.com}

\author[X. Zhu]{Xiaomeng Zhu}
 \address{\rm Department of Mathematics, Nanjing University of Aeronautics and Astronautics,
 Nanjing,  211100, P. R. China}
 \email{mooneernanjing@163.com}

\author[Q. Yue]{Qin Yue}
\address{\rm Department of Mathematics, Nanjing University of Aeronautics and Astronautics, Nanjing,  211100, P. R. China; State Key Laboratory of Cryptology, P. O. Box 5159, Beijing, 100878, P. R. China}
\email{yueqin@nuaa.edu.cn}

\subjclass[2010]{11T71, 06A11}
\keywords{few-weight codes, codes over rings, simplicial complexes, Griesmer bound}

\date{\today}


\baselineskip=20pt

\begin{abstract}   

Recently, some infinite families of binary minimal and optimal  linear codes are constructed from simplicial complexes by Hyun {\em et al}.  
Inspired by their work, we present two new constructions of codes over the ring $\Bbb F_2+u\Bbb F_2$ by employing simplicial complexes. When the simplicial complexes are all generated by a maximal element, we determine the Lee weight distributions of two classes of the codes over $\Bbb F_2+u\Bbb F_2$. Our results show that the codes have few Lee weights. Via the Gray map, we obtain an infinite family of binary codes meeting the Griesmer bound and a class of binary distance optimal  codes. 


\end{abstract}

\maketitle

\bigskip
\section{Introduction}

An $[n, k, d]$ binary linear code $\mathcal{C}$  is a $k$-dimensional subspace of $\Bbb F_2^n$ with minimum  Hamming  distance $d$.
  Let $A_i$ denote the number of codewords in $\mathcal C$
 with Hamming weight $i$. The weight enumerator of
 $\mathcal C$ is defined by
 $1+A_1z+A_2z^2+\cdots+A_nz^n.$
The sequence $(1, A_1, A_2, \ldots, A_n)$ is called the weight distribution of
 $\mathcal C$. A code $\mathcal{C}$ is $t$-weight if the number of nonzero $A_{i}$ in the sequence $(A_1, A_2, \ldots, A_n)$ is equal to $t$.

Let $\Bbb F_{q}$ be a finite field of order $q$, where $q$ is a power of a prime $p$. Let $D=\{d_{1}, d_{2}, \ldots, d_{n}\}\subseteq \Bbb F_{w}$, where $w$ is a power of $q$. Let $\Tr_{w/q}$ be the trace function from $\Bbb F_{w}$ to $\Bbb F_{q}$. A linear code of length $n$ over $\Bbb F_{q}$ is defined by
\begin{equation*}\mathcal{C}_{D}= \{(\Tr_{w/q}(xd_{1}), \ldots, \Tr_{w/q}(xd_{n})) : x\in  \Bbb F_{w}\}.\end{equation*}
 The code $\mathcal{C}_{D}$ is called a trace code over $\Bbb F_q$  and the set $D$ is called the defining set of the code $\mathcal{C}_{D}$.  Although different orderings of the elements of $D$ result in different codes, these codes are permutation equivalent and have the same parameters. If the set $D$ is properly chosen, the code $\mathcal{C}_{D}$ may have good parameters. This generic construction was first introduced by Ding {\em et al}. \cite{D1, DN}. Many known codes have been produced by selecting a proper defining set, see \cite{HY1, LYL, LM2}.   These codes can be used in secret sharing schemes \cite{ADHK,CDY} and authentication codes \cite{DHKW}.
 
 Let $R$ be a finite commutative ring, $R_m$ be  an extension of $R$ of degree $m$, and $R_m^*$ be the multiplicative group of units of $R_m$. A trace code over $R$ with a defining set $L=\{l_1,l_2,\cdots, l_n\} \subseteq R_m^*$ is defined by
 \begin{equation*}\mathcal {C}_L=\{\Tr(xl_1), \Tr(xl_2), \cdots, \Tr(xl_n)|x\in R_m\},\end{equation*}
where $\Tr(\cdot)$ is a $R$-linear function from $R_m$ to $R$. Based on the  construction above, some  codes with few weights  over rings have been obtained, see \cite{LM}, \cite{LS}-\cite{S5}.

Let  $R=\Bbb F_q+u\Bbb F_q, u^2=0$ and $\mathcal {R}=\Bbb F_w+u\Bbb F_w$. The Lee weight distribution of a trace code $\mathcal{C}_L$  has been investigated in the literature. 
 
 (1) When $R=\Bbb F_2+u\Bbb F_2, u^2=0$, $\mathcal {R}=\Bbb F_{2^m}+u\Bbb F_{2^m}$, and $L=\Bbb F_{2^m}^*+u\Bbb F_{2^m}$, the code $\mathcal{C}_L$ is a two-weight code, see \cite{S1}.
 
 
 (2) When $R=\Bbb F_p+u\Bbb F_p, u^2=0$, $\mathcal {R}=\Bbb F_{p^m}+u\Bbb F_{p^m}$, and $L=\mathcal{Q}+u\Bbb F_{p^m}$, where $p$ is an odd prime and $\mathcal{Q}$ is  the set of all square elements of $\Bbb F_{p^m}^*$,  the code $C_L$ is a two-weight or three-weight code, see \cite{S2}.
 
 (3) When $R=\Bbb F_p+u\Bbb F_p, u^2=u$, $\mathcal {R}=\Bbb F_{p^m}+u\Bbb F_{p^m}$,  $L=\mathcal{Q}+u\Bbb F_{p^m}^*$ and $L=\Bbb F_{p^m}^*+u\Bbb F_{p^m}^*$, where $\mathcal{Q}$ is  the set of all square elements of $\Bbb F_{p^m}^*$,  the code $C_L$ is a two-weight or few-weight code, see \cite{S3}.
 
 (4) When $R=\Bbb F_q+u\Bbb F_q, u^2=0$, $\mathcal {R}=\Bbb F_{q^m}+u\Bbb F_{q^m}$, and $L=C_0^{(e,q^m)}+u\Bbb F_{p^m}$, where $e$ is a divisor of $q-1$ and  $C_0^{(e,q^m)}$ is  the cyclotomic class of order $e$,  the code $C_L$ is a two-weight or few-weight code, see \cite{LM}.



Most recently,  Hyun {\em et al}.  \cite{CH, HLL} constructed some infinite families of  binary optimal and minimal  linear codes  via simplicial complexes. Inspired by their work, in this paper, we focus on few-weight codes over $\Bbb F_2+u\Bbb F_2, u^2=0$ by  employing simplicial complexes. 
Let $\Delta_1, \Delta_2$ be two simplicial complexes of $\Bbb F_2^m$ and they are not equal to $\Bbb F_2^m$ at the same time.  Let $ L_1=\Delta_1+u\Delta_2^c$ and $ L_2=\Delta_1^c+u\Delta_2^c$, where $\Delta_1^c=\Bbb F_2^m\backslash \Delta_1$ and $\Delta_2^c=\Bbb F_2^m \backslash \Delta_2$.  Two codes $\mathcal {C}_{L_1}$ and $\mathcal {C}_{L_2}$ over $\Bbb F_2+u\Bbb F_2$ are defined by \begin{equation}\mathcal {C}_{L_1}=\{c_{\mathbf{a}}=(\langle \mathbf{a}, \mathbf{l}\rangle )_{\mathbf{l}\in L_1}|\mathbf{a}\in \Bbb F_2^m+u \Bbb F_2^m\},\end{equation}
 and  \begin{equation}\mathcal {C}_{L_2}=\{c_{\mathbf{a}}=(\langle \mathbf{a}, \mathbf{l}\rangle )_{\mathbf{l}\in L_2}|\mathbf{a}\in \Bbb F_2^m+u \Bbb F_2^m\},\end{equation}
where $\langle \mathbf{a}, \mathbf{l}\rangle$ is the Euclidean inner product in $ \Bbb F_2^m+u \Bbb F_2^m$.

The rest of this paper is organized as follows. In Section 2, we recall some concepts and results. In Sections 3, we determine the Lee weight distributions of some  codes in (1.1) and (1.2).  In Section 4,  using the Gray map, we obtain two classes of binary  optimal codes and present some  examples. In Section 5, we  conclude the paper.

\section{Preliminaries}

\subsection{Simplicial complexes and generating functions}$~$

Let $\Bbb F_2$ be the finite field with order two. Assume that $m$ is a positive integer. The support $\mathrm{supp}(v)$ of a vector $v \in \Bbb F_2^m$ is defined by the set of nonzero coordinate positions. The Hamming weight $wt(v)$ of $v\in \mathbb{F}^m_2$ is defined by the size of $\mathrm{supp}(v)$. There is a bijection between $\mathbb{F}_2^m$ and $2^{[m]}$ being the power set of $[m]=\{1, \cdots, m\}$, defined by $v\mapsto$ supp$(v)$. {\bf Throughout this paper, we will identify a vector in $\mathbb{F}_2^m$  with its support.}  For two subsets $A, B\subseteq [m]$, the set $\{x: x\in A\mbox{ and } x\notin B\}$ and the number of elements in the set $A$ are denoted by   
$A\backslash B$ and $|A|$, respectively.  

For two vectors $u,v\in \mathbb{F}_2^m$, we say $v\subseteq u$ if $\mathrm{supp}(v)\subseteq \mathrm{supp}(u)$.  We say that a family $\Delta \subseteq \mathbb{F}_2^m $ is a {\bf simplicial complex} if $u\in \Delta$ and $v\subseteq u$ imply $v\in \Delta$. For a simplicial complex $\Delta$, a maximal element of $\Delta $ is one that is not properly contained in any other element of $\Delta$. Let $\mathcal{F}=\{F_1, \cdots, F_l\}$ be the family of maximal elements of $\Delta$. For each $F\subseteq [m]$, the simplicial complex $\Delta_F$ generated by $F$ is defined to be the family of all subsets of $F$.

Let $X$ be a subset of $\mathbb{F}_2^m$. Hyun {\em et al}. \cite{CH} introduced the following $m$-variable generating function associated with the set $X$:
$$\mathcal{H}_{X}(x_1,x_2\ldots, x_m)=\sum_{u\in X}\prod_{i=1}^mx_i^{u_i}\in \mathbb{Z}[x_1,x_2, \ldots, x_m],
$$
where $u=(u_1,u_2,\ldots, u_m)\in \mathbb{F}_2^m$ and $\mathbb{Z}$ is the ring of integers. We observe that

$(1)$ $\mathcal{H}_{\emptyset}(x_1,x_2\ldots, x_m)=0,$

$(2)$ $\mathcal{H}_{X}(x_1,x_2\ldots, x_m)+\mathcal{H}_{X^{c}}(x_1,x_2\ldots, x_m)=\mathcal{H}_{\mathbb{F}_2^m}(x_1,x_2\ldots, x_m)=\prod_{i\in[m]}(1+x_i)$.\\


The following lemma plays an important role in determining the Lee weight distributions of the codes defined in  (1.1) and (1.2).

\begin{lem}\cite[Theorem 1]{CH}  \label{th1}
Let $\Delta$ be a simplicial complex of $\mathbb{F}_2^n$ with the set of maximal elements $\mathcal {F}$. Then 
 \begin{align*}
 \mathcal{H}_{\Delta}(x_1,x_2\ldots, x_n)=\sum_{\emptyset\neq S\subseteq \mathcal{F}}(-1)^{|S|+1}\prod_{i\in \cap S}(1+x_i),
 \end{align*}
In particular, we have that $|\Delta|=\sum_{\emptyset\neq S\subseteq \mathcal{F}}(-1)^{|S|+1}2^{|\cap S|}$.
\end{lem}


\subsection{Lee weight and Gray map}$~$
 
 In the remainder of this paper,  we always assume that  $R=\Bbb F_2+u\Bbb F_2$, where $u^2=0$.  A linear code of $\mathcal{C} $ of length $m$ over $R$ is an $R$-submodule of $R^m$. Let $\mathbf{ x}=(x_1,x_2,\cdots, x_m)$ and $\mathbf{y}=(y_1,y_2, \cdots, y_m)$ be two vectors of $R^m$. The inner product of $\mathbf{x}$ and $\mathbf{y}$ is defined by $\langle \mathbf{x},\mathbf{y} \rangle=\sum_{i=1}^mx_iy_i\in R$. 
 
 For any $a+ub\in R$, $a,b\in \Bbb F_2$, the Gray map $\phi$ from $R$ to $\Bbb F_2^2$ is defined by $$\phi: R\to \Bbb F_2^2, a+ub\mapsto (b,a+b).$$
 Any vector $\mathbf{x}\in R^m$ can be written as $\mathbf{x}={a}+u{b}$, where ${a}, {b}\in \Bbb F_2^m$. The map $\phi$ can be extended naturally from $R^m$ to $\Bbb F_2^{2m}$ as follows:
 $$\phi: R^m\to \Bbb F_2^{2m},\mathbf{x}={a}+u{b}\mapsto ({b}, {a+b}).$$
The Hamming weight of a vector ${a}$ of length $m$ over $\Bbb F_2$ is defined to be the number of nonzero entries in the vector ${a}$. The Lee weight of a vector $\mathbf{x}$ of length $m$ over $R$ is defined to be the Hamming weight of its Gray image as follows:
$$w_L(\mathbf{x})=w_L({a}+u{b})=w_H({b})+w_H({a+b}).$$
The Lee distance of ${a,b}\in R^m$ is defined as $w_L(\mathbf{x-y})$. It is easy to check that the Gray map is an isometry from $(R^m, d_L)$ to $(\Bbb F_2^{2m}, d_H)$.
\section{The Lee weight distributions of the codes over $R$}





Let $\Delta_1, \Delta_2$ be two simplicial complexes of $\Bbb F_2^m$ and they are not equal to $\Bbb F_2^m$ at the same time.  Let $ L_1=\Delta_1+u\Delta_2^c$ and $ L_2=\Delta_1^c+u\Delta_2^c$ . We define two codes $\mathcal {C}_{L_1}$ and $\mathcal {C}_{L_2}$ over $\Bbb F_2+u\Bbb F_2$ as follows:
 \begin{equation*}\mathcal {C}_{L_1}=\{c_{\mathbf{a}}=(\langle \mathbf{a}, \mathbf{l}\rangle )_{\mathbf{l}\in L_1}|\mathbf{a}\in \Bbb F_2^m+u \Bbb F_2^m\},\end{equation*}
 and  \begin{equation*}\mathcal {C}_{L_2}=\{c_{\mathbf{a}}=(\langle \mathbf{a}, \mathbf{l}\rangle )_{\mathbf{l}\in L_2}|\mathbf{a}\in \Bbb F_2^m+u \Bbb F_2^m\}.\end{equation*}
 The lengths of the codes $\mathcal{C}_{L_1}$ and $\mathcal{C}_{L_2}$ are $|L_1|$ and $|L_2|$, respectively.  

Employing Lemma 2.1, we will give the Lee weight distributions of the codes  $\mathcal{C}_{L_1}$ and $\mathcal{C}_{L_2}$ in the case that the two  simplicial complexes  are all generated by a single maximal element.  Before giving our main results, we  prove the following lemma first.





For two subsets $X$ and $Y$ of $[m]$, we use $\chi(X|Y)$ to denote a function from $2^{[m]}\times 2^{[m]}$ to $\{0,1\}$, and  $\chi(X|Y)=1$ if and only if $X \cap Y=\emptyset$.

\begin{lem} For $A, B\subseteq [m]$, we have the following.

$(1)$ The size of the set $\{\emptyset\neq X\subseteq [m]| \chi(X| A)=1\}$ is  $2^{m-|A|}-1$; the size of the set $\{X\subseteq [m]| \chi(X| A)=0\}$ is  $2^m-2^{m-|A|}$.

$(2)$ Let $$S_1=\{\emptyset\neq X\subseteq [m]| \chi (X|A)\chi(X|B)=1 \}$$ and $$S_0=\{X\subseteq [m]| \chi (X|A)\chi(X|B)=0 \}.$$ Then $|S_1|=2^{m-|A\cup B|}-1$ and $|S_0|=2^m-2^{m-|A\cup B|}$.

$(3)$ Define $A\oplus B:= (A\cup B)\backslash (A\cap B)$. Let 
$$T_2=\{(X, Y)|\emptyset\neq X,Y\subseteq [m], X\neq Y, \chi(Y|B)=1, \chi(X | A)+\chi((X\oplus Y) | A)=2\}, $$
$$T_1=\{(X, Y)|\emptyset\neq X,Y\subseteq [m], X\neq Y, \chi(Y|B)=1,\chi(X | A)+\chi((X\oplus Y) | A)=1\}, $$
and $$T_0=\{(X, Y)|\emptyset\neq X,Y\subseteq [m], X\neq Y, \chi(Y|B)=1, \chi(X | A)+\chi((X\oplus Y) | A)=0\}.$$
Then $$|T_2|=(2^{m-|A|}-2)(2^{m-|A\cup B|}-1),|T_1|=2(2^{m-|A|}-1)(2^{m-|B|}-2^{m-|A\cup B|}),$$ and $$|T_0|=2^m(2^{m-|B|}-1)+2^{m-|A|}(1+2^{m-|A\cup B|}-2^{m+1-|B|}).$$ 
 
\end{lem}

\begin{proof} $(1)$ Each element in $[m] \backslash A$ can be in the set $X$ or not.  Hence the number of nonempty $X$  such that $\chi (X| A)=1$ is $2^{m-|A|}-1$.

$(2) $ Due to $|S_0|+|S_1|=2^m$, it suffices to determine the size of $S_1$.  Note that $\chi (X|A)\chi(X|B)=1$ if and only if  $\chi (X|A)=\chi (X|B)=1$ if and only if $X\cap (A\cup B)=\emptyset$. By (1), $|S_1|=2^{m-|A\cup B|}-1$. 

$(3)$  Notice that $T_2\cup T_1 \cup T_0=\{(X, Y)|\emptyset\neq X,Y\subseteq [m], X\neq Y, \chi(Y|B)=1\}$. Then \begin{equation}|T_2|+|T_1|+|T_0|=(2^{m-|B|}-1)(2^m-2).\end{equation} It suffices to determine the sizes of $T_2$ and $T_1$. 

 Note that $\chi(X | A)+\chi((X\oplus Y) | A)=2$ if and only if $\chi(X | A)=\chi((X\oplus Y) | A)=1$ if and only if $X\cap A=\emptyset$ and $(X\oplus Y) \cap A=\emptyset$. By definition, \begin{equation}(X\oplus Y) \cap A=((X\cup Y)\backslash (X\cap Y))\cap A=(X\cap Y^c\cap A)\cup (Y\cap X^c \cap A),\end{equation} where $X^c=[m]\backslash X$. The fact $X\cap A=\emptyset$ implies that  $X^c\cap A=A$. By (3.2),  
  $X\cap A=\emptyset$ and $(X\oplus Y) \cap A=\emptyset$ if and only if $X\cap A=\emptyset$ and $Y\cap A=\emptyset$. In a word, the element $(X,Y)$ in $T_2$ should satisfy the following conditions:
$$\emptyset\neq X, \emptyset\neq Y, X\neq Y, X\cap A=\emptyset, Y\cap (A\cup B)=\emptyset.$$
According to the fact that  $2^{m-|A|}\ge 2^{m-|A\cup B|}$, we have  \begin{equation}|T_2|=(2^{m-|A|}-2)(2^{m-|A\cup B|}-1).\end{equation}

Note that $\chi(X | A)+\chi((X\oplus Y) | A)=1$ if and only if $\chi(X | A)=1$ and $\chi((X\oplus Y) | A)=0$ or $\chi(X | A)=0$ and $\chi((X\oplus Y) | A)=1$. We divide the set $T_1$ into two parts, denoted by $T_1'$ and $T_1''$. In the following, we proceed with the proof in two cases.

$(i)$ $\chi(X | A)=1$ and $\chi((X\oplus Y) | A)=0$.  By  (3.2), $\chi((X\oplus Y) | A)=0$ if and only if $(X\oplus Y) \cap A\neq \emptyset$ if and only if $Y\cap A\neq \emptyset$. In this case, the element $(X,Y)$ in $T_1'$ should satisfy the following conditions:
$$\emptyset\neq X, X\cap A=\emptyset, Y\cap B=\emptyset, Y\cap A\neq \emptyset .$$
Then  \begin{equation}|T_1'|=(2^{m-|A|}-1) (2^{m-| B|}-2^{m-|A\cup B|}).\end{equation}

$(ii)$ $\chi(X | A)=0$ and $\chi((X\oplus Y) | A)=1$. By (3.2),  $\chi((X\oplus Y) | A)=1$ if and only if $(X\cap A)\cap Y^c=\emptyset$ and $(Y\cap A)\cap X^c=\emptyset$ if and only if $X\cap A \subseteq Y$ and $Y\cap A \subseteq X$. Namely, $X\cap A=Y\cap A$. In this case, the element $(X,Y)$ in $T_1''$ should satisfy the following conditions:
$$\emptyset\neq X, X\neq Y, Y\cap B=\emptyset, X\cap A=Y\cap A\neq\emptyset.$$
The size of $X\cap A$ can be $1, 2, \cdots, |A\backslash B|=k.$ For each $1\le i\le k$, if $|X\cap A|=i$, then there are $2^{m-|A\cup B|}C_{k}^i$ choices for $Y$ and $2^{m-|A|}-1$ choices for $X$ such that the conditions above are satisfied. Then  \begin{eqnarray} |T_1''|&=&(2^{m-|A|}-1)2^{m-|A\cup B|}\sum_{i=1}^{k} C_k^i=(2^{m-|A|}-1)2^{m-|A\cup B|}(2^{|A\backslash B|}-1)\nonumber\\
&=&(2^{m-|A|}-1)(2^{m-|B|}-2^{m-|A\cup B|}).\end{eqnarray}

By  (3.4) and (3.5),  \begin{eqnarray} |T_1|=|T_1'|+|T_1''|=2(2^{m-|A|}-1)(2^{m-|B|}-2^{m-|A\cup B|}).\end{eqnarray}

By  (3.1), (3.3), and  (3.6), \begin{eqnarray*}|T_0|&=&(2^{m-|B|}-1)(2^m-2)-(2^{m-|A|}-2)(2^{m-|A\cup B|}-1)\\
&-&2(2^{m-|A|}-1)(2^{m-|B|}-2^{m-|A\cup B|})\\
&=&2^m(2^{m-|B|}-1)+2^{m-|A|}(1+2^{m-|A\cup B|}-2^{m+1-|B|}).\end{eqnarray*}

This completes the proof.
\end{proof}

\subsection{Lee weight distribution of the code $\mathcal{C}_{L_1}$}$~$

Suppose that $\mathbf{a}=\mathbf{\alpha}+u\mathbf{\beta}$, $\mathbf{l}={t_1}+u{t_2}$, where ${\alpha=(\alpha_1, \cdots, \alpha_m),\beta=(\beta_1, \cdots, \beta_m)}\in \Bbb F_2^m$, ${t_1}\in \Delta_1$, and ${t_2}\in \Delta_2^c$. If $\mathbf{a}=\mathbf{0}$, then $w_L(c_{\mathbf{a}})=0$. Next we  assume that $\mathbf{a}\neq \mathbf{0}$. Then the Lee weight of the codeword $c_{\mathbf{a}}$ of the code $\mathcal{C}_{L_1}$ becomes that \begin{eqnarray}&&w_L(c_{\mathbf{a}})\nonumber\\
&=&w_L((\mathbf{\alpha}+u\mathbf{\beta})({t_1}+u{t_2})_{{t_1}\in \Delta_1, {t_2}\in \Delta_2^c})\nonumber\\
&=&w_L((\mathbf{\alpha}{t_1}+u(\alpha{t_2}+\beta{t_1}))_{{t_1}\in \Delta_1, {t_2}\in \Delta_2^c})\nonumber\\
&=&w_H((\alpha{t_2}+\beta{t_1})_{\mathbf{t_1}\in \Delta_1, {t_2}\in \Delta_2^c})+w_H(((\alpha+\beta){t_1}+\alpha{t_2})_{{t_1}\in \Delta_1, {t_2}\in \Delta_2^c})\nonumber\\
&=&|L_1|-\frac12\sum_{y\in\Bbb F_2}\sum_{t_1\in \Delta_1}\sum_{t_2\in \Delta_2^c}(-1)^{(\alpha{t_2}+\beta{t_1})y}\nonumber\\
&+&|L_1|-\frac12\sum_{y\in\Bbb F_2}\sum_{t_1\in \Delta_1}\sum_{t_2\in \Delta_2^c}(-1)^{((\alpha+\beta){t_1}+\alpha{t_2})y}\nonumber\\
&=&|L_1|-\frac12\sum_{t_1\in \Delta_1}(-1)^{\beta{t_1}}\sum_{t_2\in \Delta_2^c}(-1)^{\alpha{t_2}}
-\frac12\sum_{t_1\in \Delta_1}(-1)^{(\alpha+\beta)t_1}\sum_{t_2\in \Delta_2^c}(-1)^{\alpha{t_2}}\nonumber\\
&=&|L_1|-\frac12(\sum_{t_2\in \Delta_2^c}(-1)^{\alpha{t_2}})(\sum_{t_1\in \Delta_1}(-1)^{\beta{t_1}}+\sum_{t_1\in \Delta_1}(-1)^{(\alpha+\beta)t_1}).
\end{eqnarray}

\begin{thm} Let $A, B\subseteq [m]$ and $0<|B|<m$ .  Let $\Delta_A, \Delta_B$ be two simplicial complexes of $\Bbb F_2^m$ and $ L_1=\Delta_A+u\Delta_B^c$. Then the code  $\mathcal{C}_{L_1}$ has length $2^{|A|}(2^m-2^{|B|})$, size $2^{m+|A|}$, and its Lee weight distribution of  is given in Table 1.\end{thm}

\begin{table}[h]  
\caption{Lee weight distribution of the code in Theorem 3.2}   
\begin{tabu} to 0.8\textwidth{X[1.2,c]|X[2,c]}  
\hline 
\rm{Lee Weight}&\rm{Frequency}\\ 
\hline
$0$&$2^{m-|A|}$\\ 
\hline
$2^{m+|A|}$& $2^{m-|A|}(2^{m-|A\cup B|}-1)$ \\ 
\hline
$2^{|A|}(2^m-2^{|B|})$& $2^{2m}+2^{m-|A|}(2^{m-|A\cup B|}-2^{m+1-| B|})$ \\ 
\hline
$2^{m+|A|}-2^{|A|+|B|-1}$& $2^{m-|A|+1}(2^{m-|B|}-2^{m-|A\cup B|})$ \\ 
\hline
\end{tabu}  
\end{table}

\begin{proof}

It is easy to check that the length $|L_1|$ of the code $\mathcal{C}_{L_1}$ is $ 2^{|A|}(2^m-2^{|B|})$. Recall that there is a bijection between $\mathbb{F}_2^m$ and $2^{[m]}$.  For $u\in \mathbb{F}_2^m$ and  $X\subseteq\mathbb{F}_2^m$, we also use $\chi(u|X)$ to denote a Boolean function in $n$-variable, and $\chi(u|X)=1$ if and only if $u\bigcap X=\emptyset$. Suppose that $\mathbf{0}\neq\mathbf{a}=\mathbf{\alpha}+u\mathbf{\beta}$, where $\alpha=(\alpha_1, \cdots, \alpha_m)$, $\beta=(\beta_1, \cdots, \beta_m)\in \Bbb F_2^m$.
By Lemma 2.1,  \begin{eqnarray*}&&\mathcal{H}_{\Delta_A}((-1)^{\beta_1}, \cdots, (-1)^{\beta_m})= \prod _{i\in A}(1+(-1)^{\beta_i})
=\prod _{i\in A}2(1-\beta_i)=2^{|A|}\chi(\beta | A).\end{eqnarray*} 
By (3.7), 
\begin{eqnarray*}&&w_L(c_{\mathbf{a}})\nonumber\\
&=&|L_1|-\frac12( 2^m \delta_{0,\alpha}-\mathcal{H}_{\Delta_B}((-1)^{\alpha_1}, \cdots, (-1)^{\alpha_m})\mathcal{H}_{\Delta_A}((-1)^{\beta_1}, \cdots, (-1)^{\beta_m})\nonumber\\
&- &\frac12( 2^m \delta_{0,\alpha}-\mathcal{H}_{\Delta_B}((-1)^{\alpha_1}, \cdots, (-1)^{\alpha_m})\mathcal{H}_{\Delta_A}((-1)^{\alpha_1+\beta_1}, \cdots, (-1)^{\alpha_m+\beta_m})\nonumber\\
&=&  |L_1|-2^{|A|-1}(2^m\delta_{0,\alpha}-2^{|B|}\chi(\alpha | B))(\chi(\beta | A)+\chi(\alpha+\beta | A) ).
\end{eqnarray*} where $\delta$ is the Kronecker delta function.

Suppose that $\mathrm{supp}(\alpha)=X$ and $\mathrm{supp}(\beta)=Y$.  It is easy to verify that $\mathrm{supp}(\alpha+\beta)=X\oplus Y$, which is defined in Lemma 3.1.
Next we divide the proof into two cases.

$(1)$  $\alpha=0$. Then 
$$w_L(c_{\mathbf{a}})=|L_1|-2^{|A|}(2^m-2^{|B|})(\chi(\beta | A) .$$ In this case, $w_L(c_{\mathbf{a}})=0$ or $2^{|A|}(2^m-2^{|B|})$ and the frequencies are given in Lemma 3.1 (1).


$(2)$ $\alpha\neq 0$.  Then $$w_L(c_{\mathbf{a}})=|L_1|+2^{|A|+|B|-1}\chi(\alpha |B)(\chi(\beta |A)+\chi(\alpha+\beta |A)).$$ By Lemma 3.1, the number of $\alpha$ such that $\chi(\alpha | B)=0$ is  $2^m-2^{m-|B|}$,  in this case we have  $w_L(c_{\mathbf{a}})=2^{|A|}(2^m-2^{|B|})$.   Let $$P_2=\{(\alpha, \beta)|\alpha, \beta \in \Bbb F_2^{m}, \chi(\alpha | B)=1, \chi(\beta | A)+\chi(\alpha+\beta | A)=2\}, $$
$$P_1=\{(\alpha, \beta)|\alpha, \beta \in \Bbb F_2^{m},  \chi(\alpha | B)=1, \chi(\beta | A)+\chi(\alpha+\beta | A)=1\}, $$
and $$P_0=\{(\alpha, \beta)|\alpha, \beta \in \Bbb F_2^{m}, \chi(\alpha | B)=1, \chi(\beta | A)+\chi(\alpha+\beta | A)=0\}.$$Then the number of $\mathbf{a}$ with $w_L(c_{\mathbf{a}})=2^{m+|A|}$ is $|P_2|$;  the number of $\mathbf{a}$ with $w_L(c_{\mathbf{a}})=2^{m+|A|}-2^{|A|+|B|-1}$ is $|P_1|$; the number of $\mathbf{a}$ with $w_L(c_{\mathbf{a}})=2^{|A|}(2^m-2^{|B|})$ is $|P_0|$. 

By the proof of Lemma 3.1,  $$|P_2|=2^{m-|A|}(2^{m-|A\cup B|}-1),|P_1|=2^{m-|A|+1}(2^{m-|B|}-2^{m-|A\cup B|}),$$ and $$|P_0|=2^m(2^{m-|B|}-1)+2^{m-|A|}(1+2^{m-|A\cup B|}-2^{m+1-|B|}).$$





This completes the proof.
\end{proof}

\begin{rem} In Theorem 3.2, if $A\cup B=[m]$ or $A\subseteq B$, then the code $\mathcal{C}_{L_1}$ is a two-Lee-weight code.

\end{rem}

\begin{cor}   Suppose that $\Delta_B$ is a simplicial complex with a single maximal element $ B \subseteq [m]$ with $0<|B|<m$.   If $ L_1=u\Delta_B^c$, then the code $\mathcal{C}_{L_1}$ has length $2^m-2^{|B|}$, size $2^{m}$, and its Lee weight distribution of  is given in Table 2.

\end{cor}

\begin{table}[h]  
\caption{Lee weight distribution of the code in Corollary 3.4}   
\begin{tabu} to 0.5\textwidth{X[1,c]|X[2,c]}  
\hline 
\rm{Lee Weight}&\rm{Frequency}\\ 
\hline
$0$&$2^{m}$\\ 
\hline
$2^{m}$&$2^m(2^{m-|B|}-1)$\\
\hline
$2^{m}-2^{|B|}$& $2^{2m}-2^{2m-|B|}$ \\ 
\hline
\end{tabu}  
\end{table}

\begin{cor}   Suppose that $\Delta_B$ is a simplicial complex with a single maximal element $ B \subseteq [m]$ with $0<|B|<m$.   If $ L_1=\Bbb F_2^m+u\Delta_B^c$, then the code $\mathcal{C}_{L_1}$ has length $2^m(2^m-2^{|B|})$, size $2^{2m}$, and its Lee weight distribution   is given in Table 3.

\end{cor}

\begin{table}[h]  
\caption{Lee weight distribution of the code in Corollary 3.5}   
\begin{tabu} to 0.6\textwidth{X[1,c]|X[2,c]}  
\hline 
\rm{Lee Weight}&\rm{Frequency}\\ 
\hline
$0$&$1$\\ 
\hline
$2^{m}(2^m-2^{|B|})$&$2^{2m}-2^{m+1-|B|}+1$\\
\hline
$2^{2m}-2^{m-1+|B|}$& $2^{m+1-|B|}-2$ \\ 
\hline
\end{tabu}  
\end{table}

\subsection{Lee weight distributions of the code $\mathcal{C}_{L_2}$}$~$

Suppose that $\mathbf{a}=\mathbf{\alpha}+u\mathbf{\beta}$, $\mathbf{l}={t_1}+u{t_2}$, where ${\alpha=(\alpha_1, \cdots, \alpha_m),\beta=(\beta_1, \cdots, \beta_m)}\in \Bbb F_2^m$, ${t_1}\in \Delta_1^c$, and ${t_2}\in \Delta_2^c$. If $\mathbf{a}=\mathbf{0}$, then $w_L(c_{\mathbf{a}})=0$. Next we  assume that $\mathbf{a}\neq \mathbf{0}$.  By  (3.7), then the Lee weight of the codeword $c_{\mathbf{a}}$ of the code $\mathcal{C}_{L_2}$ becomes that 
\begin{eqnarray}&&w_L(c_{\mathbf{a}})\nonumber\\
&=&w_H((\alpha{t_2}+\beta{t_1})_{\mathbf{t_1}\in \Delta_1^c, {t_2}\in \Delta_2^c})+w_H(((\alpha+\beta){t_1}+\alpha{t_2})_{{t_1}\in \Delta_1^c, {t_2}\in \Delta_2^c})\nonumber\\
&=&|L_2|-\frac12\sum_{y\in\Bbb F_2}\sum_{t_1\in \Delta_1^c}\sum_{t_2\in \Delta_1^c}(-1)^{(\alpha{t_2}+\beta{t_1})y}\nonumber\\
&+&|L_2|-\frac12\sum_{y\in\Bbb F_2}\sum_{t_1\in \Delta_1^c}\sum_{t_2\in \Delta_2^c}(-1)^{((\alpha+\beta){t_1}+\alpha{t_2})y}\nonumber\\
&=&|L_2|-\frac12\sum_{t_1\in \Delta_1^c}(-1)^{\beta{t_1}}\sum_{t_2\in \Delta_2^c}(-1)^{\alpha{t_2}}
-\frac12\sum_{t_1\in \Delta_1^c}(-1)^{(\alpha+\beta)t_1}\sum_{t_2\in \Delta_2^c}(-1)^{\alpha{t_2}}\nonumber\\
&=&|L_2|-\frac12(\sum_{t_2\in \Delta_2^c}(-1)^{\alpha{t_2}})(\sum_{t_1\in \Delta_1^c}(-1)^{\beta{t_1}}-\sum_{t_1\in \Delta_1^c}(-1)^{(\alpha+\beta)t_1}).
\end{eqnarray}



\begin{thm} Let $A, B\subseteq [m]$, $0<|A|<m$, and $0<|B|<m$ .  Let $\Delta_A, \Delta_B$ be two simplicial complexes of $\Bbb F_2^m$ and $ L_2=\Delta_A^c+u\Delta_B^c$. Then the code  $\mathcal{C}_{L_2}$ has length $(2^m-2^{|A|})(2^m-2^{|B|})$, size $2^{2m}$, and its Lee weight distribution   is given in Table 4.

\end{thm}

\begin{table}[h]  
\caption{Lee weight distribution of the code in Theorem 3.6}   
\begin{tabu} to 1\textwidth{X[2,c]|X[2,c]}  
\hline 
\rm{Lee Weight}&\rm{Frequency}\\ 
\hline
$0$&$1$\\ 
\hline
$2^{m}(2^m-2^{|B|})$& $2^{m-|B|}-1$ \\ 
\hline
$(2^m-2^{|A|})(2^m-2^{|B|})+2^{|B|-1}(2^m-2^{|A|+1})$& $2(2^{m-|A\cup B|}-1)$ \\ 
\hline
$(2^m-2^{|A|})(2^m-2^{|B|})+2^{|B|-1}(2^m-2^{|A|})$& $2(2^{m-|B|}-2^{m-|A\cup B|})$ \\ 
\hline
$(2^m-2^{|A|})(2^m-2^{|B|})-2^{|A|+|B|}$& $(2^{m-|A|}-2)(2^{m-|A\cup B|}-1)$ \\ 
\hline
$(2^m-2^{|A|})(2^m-2^{|B|})-2^{|A|+|B|-1}$& $2(2^{m-|A|}-1)(2^{m-|B|}-2^{m-|A\cup B|})$ \\ 
\hline
$(2^m-2^{|A|})(2^m-2^{|B|})$&$2^{2m}+2^{m-|A|}(1+2^{m-|A\cup B|}-2^{m+1-|B|})-2^{m-|B|}$\\
\hline
\end{tabu}  
\end{table}


\begin{proof}  The length $|L_2|$ of the code $\mathcal{C}_{L_2}$ is $ (2^m-2^{|A|})(2^m-2^{|B|})$. Suppose that $\mathbf{0}\neq\mathbf{a}=\mathbf{\alpha}+u\mathbf{\beta}$, where $\alpha=(\alpha_1, \cdots, \alpha_m)$, $\beta=(\beta_1, \cdots, \beta_m)\in \Bbb F_2^m$.
By  (3.8), 
\begin{eqnarray*}&&w_L(c_{\mathbf{a}})\nonumber\\
&=&|L_2|-\frac12( 2^m \delta_{0,\alpha}-\mathcal{H}_{\Delta_B}((-1)^{\alpha_1}, \cdots, (-1)^{\alpha_m})( 2^m \delta_{0,\beta}-\mathcal{H}_{\Delta_A}((-1)^{\beta_1}, \cdots, (-1)^{\beta_m})\nonumber\\
&- &\frac12( 2^m \delta_{0,\alpha}-\mathcal{H}_{\Delta_B}((-1)^{\alpha_1}, \cdots, (-1)^{\alpha_m})( 2^m \delta_{0,\alpha+\beta}-\mathcal{H}_{\Delta_A}((-1)^{\alpha_1+\beta_1}, \cdots, (-1)^{\alpha_m+\beta_m})\nonumber\\
&=&  |L_2|-\frac12(2^m\delta_{0,\alpha}-2^{|B|}\chi(\alpha | B))(2^m\delta_{0,\beta}-2^{|A|}\chi(\beta | A))\\
&-&\frac12(2^m\delta_{0,\alpha}-2^{|B|}\chi(\alpha | B))(2^m\delta_{0,\alpha+\beta}-2^{|A|}\chi(\alpha+\beta | A)).
\end{eqnarray*} 

Suppose that $\mathrm{supp}(\alpha)=X$ and $\mathrm{supp}(\beta)=Y$.  It is easy to verify that $\mathrm{supp}(\alpha+\beta)=X\oplus Y$, which is defined in Lemma 3.1.
Next we divide the proof into four cases.

$(1)$  $\alpha=0$ and $\beta\neq 0$. Then 
$$w_L(c_{\mathbf{a}})=|L_2|+(2^m-2^{|B|})2^{|A|}\chi(\beta | A).$$ In this case, $w_L(c_{\mathbf{a}})=(2^m-2^{|A|})(2^m-2^{|B|})$ or $2^m(2^m-2^{|B|})$ and the frequencies are given in Lemma 3.1 (1).


$(2)$ $\alpha\neq 0$ and $\beta= 0$.  Then $$w_L(c_{\mathbf{a}})=|L_2|+2^{|B|-1}\chi(\alpha | B)(2^m-2^{|A|}-2^{|A|}\chi(\alpha | A)).$$ By Lemma 3.1, the number of $\alpha$ such that $\chi(\alpha | B)=0$ is  $2^m-2^{m-|B|}$,  in this case we have  $w_L(c_{\mathbf{a}})=(2^m-2^{|A|})(2^m-2^{|B|})$.  By Lemma 3.1, the number of $\alpha$ such that  $\chi(\alpha | B)\chi(\alpha | A)=1$ is $2^{m-|A\cup B|}-1$, and in this case we have $w_L(c_{\mathbf{a}})=(2^m-2^{|A|})(2^m-2^{|B|})+2^{|B|-1}(2^m-2^{|A|+1}) $.  The number of $\alpha$ such that $\chi(\alpha | B)=1$ and $\chi(\alpha | A)=0$ is $2^{m-|B|}-1-(2^{m-|A\cup B|}-1)=2^{m-|B|}-2^{m-|A\cup B|}$, and in this case we have $w_L(c_{\mathbf{a}})=(2^m-2^{|A|})(2^m-2^{|B|})+2^{|B|-1}(2^m-2^{|A|}) =(2^m-2^{|A|})(2^m-2^{|B|-1})$.

$(3)$ $\alpha=\beta\neq 0$.  Then $$w_L(c_{\mathbf{a}})=|L_2|+2^{|B|-1}\chi(\alpha | B)(2^m-2^{|A|}-2^{|A|}\chi(\alpha | A)).$$ Similar  to (2), we have the Lee weights and their  frequencies.


$(4)$  $\alpha\neq 0$, $\beta\neq 0$, and $\alpha\neq \beta$.
 Then $$w_L(c_{\mathbf{a}})=|L_2|-2^{|A|+|B|-1}\chi(\alpha | B)(\chi(\beta | A)+\chi(\alpha+\beta |A)). $$ Hence the number of $\alpha$ such that $\chi(\alpha | B)=0$ is  $2^m-2^{m-|B|}$,  in this case we have  $w_L(c_{\mathbf{a}})=(2^m-2^{|A|})(2^m-2^{|B|})$.  Let $$T_2=\{(\alpha, \beta)|\alpha, \beta \in \Bbb F_2^{m*}, \alpha\neq \beta, \chi(\alpha | B)=1, \chi(\beta | A)+\chi(\alpha+\beta | A)=2\}, $$
$$T_1=\{(\alpha, \beta)|\alpha, \beta \in \Bbb F_2^{m*}, \alpha\neq \beta, \chi(\alpha | B)=1, \chi(\beta | A)+\chi(\alpha+\beta | A)=1\}, $$
and $$T_0=\{(\alpha, \beta)|\alpha, \beta \in \Bbb F_2^{m*}, \alpha\neq \beta, \chi(\alpha | B)=1, \chi(\beta | A)+\chi(\alpha+\beta | A)=0\}.$$Then the number of $\mathbf{a}$ with $w_L(c_{\mathbf{a}})=(2^m-2^{|A|})(2^m-2^{|B|})-2^{|A|+|B|}=2^m(2^m-2^{|A|}-2^{|B|})$ is $|T_2|$;  the number of $\mathbf{a}$ with $w_L(c_{\mathbf{a}})=(2^m-2^{|A|})(2^m-2^{|B|})-2^{|A|+|B|-1}$ is $|T_1|$; the number of $\mathbf{a}$ with $w_L(c_{\mathbf{a}})=(2^m-2^{|A|})(2^m-2^{|B|})$ is $|T_0|$. 

By Lemma 3.1,  $$|T_2|=(2^{m-|A|}-2)(2^{m-|A\cup B|}-1),|T_1|=2(2^{m-|A|}-1)(2^{m-|B|}-2^{m-|A\cup B|}),$$ and $$|T_0|=2^m(2^{m-|B|}-1)+2^{m-|A|}(1+2^{m-|A\cup B|}-2^{m+1-|B|}).$$





This completes the proof.
\end{proof}

\begin{rem} In Theorem 3.6, if $A\cup B=[m]$ or $A\subseteq B$, then the code $\mathcal{C}_{L_2}$ is a four-Lee-weight code.

\end{rem}

\begin{thm}   Suppose that $\Delta_A $ is a simplicial complex with a single maximal element $ A \subseteq [m]$ with $0<|A|<m$.   If $ L_2=\Delta_A^c$, then the code  $\mathcal{C}_{L_2}$ has length $2^m-2^{|A|}$, size $2^{2m}$, and its Lee weight distribution   is given in Table 5.
\end{thm}
\begin{table}[h]  
\caption{Lee weight distribution of the code in Theorem 3.8}   
\begin{tabu} to 0.6\textwidth{X[1,c]|X[2,c]}  
\hline 
\rm{Lee Weight}&\rm{Frequency}\\ 
\hline
$0$&$1$\\ 
\hline
$2^{m}$&$(2^{m-|A|}-1)^2$\\
\hline
$2^{m}-2^{|A|}$& $(2^m-2^{m-|A|})^2$ \\ 
\hline
$2^{m-1}$& $2(2^{m-|A|}-1)$ \\ 
\hline
$2^{m-1}-2^{|A|-1}$& $2(2^m-2^{m-|A|})$ \\ 
\hline
$2^{m}-2^{|A|-1}$& $2(2^m-2^{m-|A|})(2^{m-|A|}-1)$ \\ 
\hline
\end{tabu}  
\end{table}

\begin{proof} The length of the code $\mathcal{C}_{L_2}$ is $|L_2|=2^m-2^{|A|}$.  
By (3.8), \begin{eqnarray*}w_L(c_{\mathbf{a}})&=&|L_2|-\frac12( 2^m \delta_{0,\beta}-\mathcal{H}_{\Delta_A}((-1)^{\beta_1}, \cdots, (-1)^{\beta_m})\nonumber\\
&- &\frac12( 2^m \delta_{0,\alpha+\beta}-\mathcal{H}_{\Delta_A}((-1)^{\alpha_1+\beta_1}, \cdots, (-1)^{\alpha_m+\beta_m})\nonumber\\
&=&|L_2|-\frac12( 2^m \delta_{0,\beta}-2^{|A|}\chi(\beta | A))- \frac12( 2^m \delta_{0,\alpha+\beta}-2^{|A|}\chi(\beta | A)).
\end{eqnarray*}
Next we proceed the proof by four cases.

$(1)$  $\alpha=0$ and $\beta\neq 0$. Then \begin{eqnarray*}w_L(c_{\mathbf{a}})
&=&2^m-2^{|F|} +2^{|A|}\chi(\beta | A).\end{eqnarray*} By Lemma 3.1, the number of $\beta$ such that $\chi(\beta |A)=1$ is  $2^{m-|A|}-1$,  in this case we have  $w_L(c_{\mathbf{a}})=2^m$. On the other hand, the number of $\beta$ such that $\chi(\beta |A)=1$ is  $2^m-2^{m-|A|}$,  in this case we have  $w_L(c_{\mathbf{a}})=2^m-2^{|A|}$.

$(2)$  $\beta=0$ and $\alpha\neq 0$.  Then \begin{eqnarray*}w_L(c_{\mathbf{a}})
&=&2^{m-1} -2^{|A|-1} + 2^{|A|-1}\chi(\alpha | A).\end{eqnarray*} By Lemma 3.1,  the number of $\alpha$ such that $\chi(\alpha |A)=1$ is  $2^{m-|A|}-1$,  in this case we have  $w_L(c_{\mathbf{a}})=2^{m-1}$. On the other hand, the number of $\alpha$ such that $\chi(\alpha |A)=1$ is  $2^m-2^{m-|A|}$,  in this case we have  $w_L(c_{\mathbf{a}})=2^{m-1}-2^{|A|-1}$.

$(3)$  $\alpha=\beta\neq 0$. Then \begin{eqnarray*}w_L(c_{\mathbf{a}})
&=&2^{m-1} -2^{|A|-1} + 2^{|A|-1}\chi(\alpha | A).\end{eqnarray*} Similar to (2), we have the Lee weights and their frequencies.

$(4)$ $\alpha\neq 0$, $\beta\neq 0$, and $\alpha\neq\beta$.  Then \begin{eqnarray*}w_L(c_{\mathbf{a}})
&=&2^m-2^{|A|}+ 2^{|A|-1}(\chi(\alpha | A)+\chi(\alpha+\beta | A)).\end{eqnarray*}
Let $$T_2=\{(\alpha, \beta)|\alpha, \beta \in \Bbb F_2^{m*}, \alpha\neq \beta, \chi(\alpha | A)+\chi(\alpha+\beta | A)=2\}, $$
$$T_1=\{(\alpha, \beta)|\alpha, \beta \in \Bbb F_2^{m*}, \alpha\neq \beta, \chi(\alpha | A)+\chi(\alpha+\beta | A)=1\}, $$
and $$T_0=\{(\alpha, \beta)|\alpha, \beta \in \Bbb F_2^{m*}, \alpha\neq \beta, \chi(\alpha | A)+\chi(\alpha+\beta | A)=0\}.$$ Then the number of $\mathbf{a}$ with $w_L(c_{\mathbf{a}})=2^m$ is $|T_2|$;  the number of $\mathbf{a}$ with $w_L(c_{\mathbf{a}})=2^m-2^{|A|-1}$ is $|T_1|$; the number of $\mathbf{a}$ with $w_L(c_{\mathbf{a}})=2^m-2^{|A|}$ is $|T_0|$.

 Suppose that $\mathrm{supp}(\alpha)=X$ and $\mathrm{supp}(\beta)=Y$. Taking $B=\emptyset$ in Lemma 3.1 (3), 
$$|T_2|=(2^{m-|A|}-1)(2^{m-|A|}-2), |T_1|=2(2^{m-|A|}-1)(2^m-2^{m-|A|}),$$
and $$|T_0|=(2^m-2^{m-|A|})(2^m-2^{m-|A|}-1).$$

This complete the proof.
\end{proof}

\begin{cor} If $|A|=m-1$, then the code $\mathcal{C}_{L_2}$ in Theorem 3.8 is a four-Lee-weight code and its Lee weight distribution is given in Table 6.

\end{cor}

\begin{table}[h]  
\caption{Lee weight distribution of the code in Corollary 3.9}   
\begin{tabu} to 0.5\textwidth{X[1,c]|X[2,c]}  
\hline 
\rm{Lee Weight}&\rm{Frequency}\\ 
\hline
$0$&$1$\\ 
\hline
$2^{m}$&$1$\\
\hline
$2^{m-1}$& $2^{2m}-2^{m+2}+6$ \\ 
\hline
$2^{m-2}$& $2^{m+1}-4$ \\ 
\hline
$2^m-2^{m-2}$& $2^{m+1}-4$ \\ 
\hline
\end{tabu}  
\end{table}


\section{Optimal binary codes and examples}

Recall that the Gray map $\phi$  in Section 2 is an isometry from $(R^m, d_L)$ to $(\Bbb F_2^{2m}, d_H)$. Then we should consider the binary codes $\phi(\mathcal{C}_{L_1})$ and $\phi(\mathcal{C}_{L_2})$ presented in Section 3. In this section, we will present some binary optimal codes and  numeral examples.

An $ [n,k,d]$ code  $\mathcal{C}$ is called distance optimal if no $[n,k,d+1]$ code exists, and is called almost optimal if the code $[n, k, d + 1]$ is optimal, see \cite[Chapter 2]{HP}.  For an $[n, k, d]$ binary  code, the following well-known bound is called the Griesmer bound, see \cite{G},  \begin{eqnarray*} \sum_{i=0}^{k-1}\bigg\lceil {\frac{d}{2^i}} \bigg\rceil \le n,  \end{eqnarray*} where $\lceil {x} \rceil$ denotes the smallest integer greater than or equal to $x$. 

\begin{thm} The code $\phi(\mathcal{C}_{L_1})$ in Theorem 3.2 is distance optimal.
\end{thm}

\begin{proof} By Theorem 3.2, the code $\phi(\mathcal{C}_{L_1})$ has the following parameters:
$$n=2^{|A|+1}(2^m-2^{|B|}), k=m+|A|, d=2^{|A|}(2^m-2^{|B|}).$$ Then \begin{eqnarray*}&&\sum_{i=0}^{m+|A|-1}\bigg\lceil {\frac{2^{m+|A|}-2^{|A|+|B|}+1}{2^i}} \bigg\rceil\nonumber\\
&=&(2^{m+1+|A|}-2)-(2^{|A|+|B|+1}-1)+1+|A|+|B|\nonumber\\
&=&2^{|A|+1}(2^m-2^{|B|})+|A|+|B|\\
&>&2^{|A|+1}(2^m-2^{|B|}).
\end{eqnarray*}
By the Griesmer bound, there is no code with parameters $[n,k,d+1]$.

This completes the proof.
\end{proof}

We give the following example to illustrate Theorem 3.2.

\begin{exa} Suppose that $m=2$ in Theorem 3.2. 


$(1)$ If $A=B$ and $|A|=|B|=1$, then  $\phi(\mathcal{C}_{L_1})$ is a two-weight binary code with  parameters $[8, 3, 4]$ and weight enumerator $1+6z^{4}+z^{8}.$


$(2)$ If $|A|=2, $ and $|B|=1$, then   $\phi(\mathcal{C}_{L_1})$ is a two-weight binary code with  parameters $[16, 4, 8]$ and  weight enumerator $1+13z^{8}+2z^{12}.$
\end{exa}

\begin{thm} In Theorem 3.6, if $A=B$ and $|A|=m-1$, then the code $\phi(\mathcal{C}_{L_2})$  meets the Griesmer bound with equality.
\end{thm}

\begin{proof} By Theorem 3.6, if $A=B$ and $|A|=m-1$, then the code $\phi(\mathcal{C}_{L_2})$ has the following parameters:
$$n=2^{2m-1}, k=2m, d=2^{2m-2}.$$ Therefore \begin{eqnarray*}&&\sum_{i=0}^{2m-1}\bigg\lceil {\frac{2^{2m-2}}{2^i}} \bigg\rceil
=\sum_{i=0}^{2m-2}\bigg\lceil {\frac{2^{2m-2}}{2^i}}\bigg\rceil+1
=(2^{2m-1}-1)+1
=2^{2m-1}.
\end{eqnarray*}

This completes the proof.
\end{proof}




The following are some  examples.

\begin{exa} Suppose that $m=3$ in Theorem 3.6. 



$(1)$ If  $|A|=1$, $|B|=2$ and $A\cap B=\emptyset$, then  $\phi(\mathcal{C}_{L_2})$ is a four-weight binary code with  parameters $[48, 6, 20]$ and  weight enumerator $1+6z^{20}+54z^{24}+z^{32}+2z^{36}.$



$(2)$ If   $|A|=|B|=2$ and $A=B$,  then  $\phi(\mathcal{C}_{L_2})$ is a two-weight binary code with  parameters $[32, 6, 16]$ and weight enumerator $1+62z^{16}+z^{32}.$ In fact, the code $\phi(\mathcal{C}_{L_2})$ meets the Griesmer bound with equality.
\end{exa}

\begin{exa} Suppose that $m=3$ in Corollary 3.9. 

$(1)$ If $|A|=1$,  then  $\phi(\mathcal{C}_{L_2})$ is a five-weight binary code with parameters $[12, 6, 3]$ and weight enumerator $1+8z^{3}+6z^{4}+16z^{6}+24z^{7}+9z^{8}.$ In fact,  $\phi(\mathcal{C}_{L_2})$ is almost optimal according to  \cite{G2}.

$(2)$ If $|A|=2$,  then  $\phi(\mathcal{C}_{L_2})$ is a four-weight binary code with  parameters $[8, 6, 2]$ and weight enumerator $1+12z^{2}+38z^{4}+12z^{6}+z^{8}.$
In fact,  $\phi(\mathcal{C}_{L_2})$ is distance optimal according to  \cite{G2}.
\end{exa}





\section{Concluding remarks}
The main contributions of this paper are  the following

\begin{itemize}
\item Two constructions of  codes  over $\Bbb F_2+u\Bbb F_2$, where $u^2=0$, defined in  (1.1) and (1.2)  associated with simplicial complexes.

\item The determination of the Lee weight distributions of two classes of the codes over $\Bbb F_2+u\Bbb F_2$ when these simplicial complexes are all generated by a single maximal element (Theorems 3.2, 3.4, 3.6, and 3.8).

\item A class of binary distance optimal codes (Theorem 4.1) and an infinite family of  binary optimal   codes meeting the  Griesmer bound (Theorem 4.3).
\end{itemize}

 It is worthy to note that in \cite{LM,S2} the authors obatined a class of two-Lee-weight  codes over $\Bbb F_2+u\Bbb F_2$ with parameters $[2^{2m}-2^m,2m]$ and Lee weight enumerator $1+(2^{2m}-2^m)z^{2^{2m}-2^m}+(2^m-1)z^{2^{2m}}$.
 It is easy to verify that the above result is an immediate consequence of Corollary 3.4. To the best of our knowledge, the two-Lee-weight  codes presented in Remark 3.3 and Corollaries 3.4 and 3.5 are new and have flexible parameters.

By massive computation, some binary optimal  codes can be also  found from Theorems 3.2 and 3.6. Very recently, Hyun et al. extended the construction of linear codes to posets in \cite{HKWY}.
It would be interesting to find more binary optimal   codes by employing  posets.

\bigskip
\section*{Acknowledgments}
This research was supported by the National Natural Science Foundation of China (No. 61772015), the Foundation of Science and Technology on Information Assurance Laboratory (No. KJ-17-010).  The authors are very grateful to the reviewers and the Associate Editor Prof. G. Matthews for their valuable comments and suggestions
to improve the quality of this paper.

\end{document}